%% file: main.tex
\documentclass[11pt,letterpaper]{article}

\usepackage[bottom=1in,top=1in,left=1in,right=1in]{geometry}

\usepackage[utf8]{inputenc}
\usepackage[english]{babel}

\usepackage{latexsym, amssymb, bbm}
\usepackage{amsmath,amsthm,amssymb}
\usepackage{breqn}
\usepackage{verbatim}
\usepackage{fancyhdr}
\usepackage{enumitem}
\usepackage{tikz}
\usepackage{algorithm}
\usepackage{algpseudocode}
\usepackage{xspace}
\usepackage{nicefrac}

\makeatletter
\def\BState{\State\hskip-\ALG@thistlm}
\makeatother

\usepackage{tikz}
\usepackage{pgfplots}

\usepackage{ifpdf}
\newcommand{\mydriver}{hypertex}
\ifpdf
 \renewcommand{\mydriver}{pdftex}
\fi
\definecolor{Darkblue}{rgb}{0.1,0,0.8}
\definecolor{Brown}{cmyk}{0,0.81,1.,0.60}
\definecolor{Purple}{cmyk}{0.45,0.86,0,0}
\usepackage[breaklinks,\mydriver]{hyperref}
\hypersetup{colorlinks=true,
            citebordercolor={.6 .6 .6},linkbordercolor={.6 .6 .6},
citecolor=blue,urlcolor=black,linkcolor=blue,pagecolor=black}

\usepackage[colorinlistoftodos,textsize=tiny,textwidth=2cm,color=red!25!white,obeyFinal]{todonotes}

\newtheorem{theorem}{Theorem}[section]
\newtheorem{lemma}[theorem]{Lemma}

\newtheorem{claim}[theorem]{Claim}

\theoremstyle{definition}

\theoremstyle{remark}

\algrenewcommand\algorithmiccomment[2][\normalsize]{{#1\hfill\(\triangleright\) \emph{#2}}}

\DeclareMathOperator{\dist}{dist}

\newcommand{\curr}{\ensuremath{K_{curr}}}
\newcommand{\currA}{\curr^A}
\renewcommand{\emptyset}{\varnothing}

\renewcommand{\part}[1] {\vspace{.10in} {\bf (#1)}}

\input{def}

\newcommand{\Tighten}{\textsf{Tighten}\xspace}
\newcommand{\Bounded}{\textsf{Bounded}\xspace}

\newcommand{\cbc}{convex body chasing\xspace}

\newcommand{\defeq}{:=}

\newcommand{\sse}{\subseteq}

\begin{document}
\title{A Nearly-Linear Bound for Chasing Nested Convex
  Bodies\footnote{This work was done while Michael B.\ Cohen and Yin Tat
    Lee were at Microsoft Research. Part of this research was done at
    the Simons Institute for the Theory of Computing.}}

\author{C.J.\ Argue \and S\'ebastien Bubeck \and Michael B.\ Cohen \and
  Anupam Gupta \and Yin Tat Lee}

\maketitle
\thispagestyle{empty}

\begin{abstract}
 Friedman and Linial \cite{FL93} introduced the \cbc problem to
explore the interplay between geometry and competitive ratio in metrical task systems.
In \cbc, at each time step $t \in \N$, the online algorithm receives a request in the form
of a convex body $K_t \subset \R^d$ and must output a point $x_t \in K_t$.
The goal is to minimize the total movement between consecutive output points,
where the distance is measured in some given norm.

This problem is still far from being understood. Recently Bansal et
al.~\cite{Bansal} gave an $6^d(d!)^2$-competitive algorithm for the
  \emph{nested} version, where each convex body is
  contained within the previous one. We propose a different
 strategy which is $O(d \log d)$-competitive algorithm for this nested
 \cbc problem. Our algorithm works for any norm.
 This result is almost tight, given an $\Omega(d)$ lower bound for the $\ell_{\infty}$ norm~\cite{FL93}.
\end{abstract}  

\section{Introduction}
\label{sec:introduction}

We consider the \emph{\cbc} problem. This is a problem in online
algorithms: the input is a starting point $x_0\in \R^d$, and the
request sequence consists of convex bodies
$K_1,K_2,\dots, K_t, \dots \subseteq \R^d$. Upon seeing request $K_t$ (but before seeing $K_{t+1}$) we must choose a feasible point
$x_t\in K_t$ within this body. The cost at the $t^{th}$ step is $\|x_t-x_{t-1}\|$, the distance moved at this time step, where $\|\cdot\|$ is some fixed norm. The objective is to minimize the total distance moved
by the algorithm:
\[ \sum_{t=0}^T \|x_{t+1} - x_t\| . \]
We consider the model of
\emph{competitive analysis}, i.e., we want to bound the worst-case ratio between the algorithm's
cost and the optimal cost to serve the request sequence.

\newcommand{\di}{\mathrm{dist}}

This problem was introduced by Friedman and Linial~\cite{FL93}, with the
goal of understanding the interplay between geometry and competitive
ratio in online problems on metric spaces, many of which can be modeled
using the \emph{metrical task systems} framework (MTS)~\cite{BRS}. In
MTS one considers an arbitrary metric space $(X,\di)$, and the request sequence
is given by functions $f_t: X \to \R_+$ which specify the cost to serve the
$t^{th}$ request at each point in $X$. The goal is to minimize the movement
plus service costs. If $(X,\di) = (\R^d, \|\cdot\|)$ and the functions $f_t$
are zero within the convex body $K_t$ (and $\infty$ outside), one gets
the \cbc problem.

The role of the metric geometry is poorly understood both for the general MTS
problem, as well as for the special case of \cbc. The latter
problem is trivial for $d=1$, and Friedman and Linial gave a competitive algorithm
for $d = 2$. However currently there is no known algorithm with a
finite competitive ratio for $d \geq 3$. Friedman and Linial gave a lower bound of
$\Omega(d^{1-1/p})$ when $\|\cdot\| = \|\cdot\|_p$ (the $\ell_p^d$ norm) based on chasing faces of the hypercube; and for
$p=1$ a lower bound of $\Omega(\log d)$ follows from~\cite{BN-mono} (both results hold in the nested version of the problem). We
also know positive results for general $d$ when the convex bodies are
lower-dimensional objects (e.g., lines or planes or affine
subspaces)~\cite{FL93, Sitters, Anto}, but these ideas do not seem to
generalize to chasing full-dimensional objects.

In this paper we restrict to \emph{nested} instances of \cbc, where the
bodies are contained within each other, i.e.,
$K_1 \supseteq K_2\supseteq K_3\supseteq \dots$. In this case, the
optimal offline algorithm at time $T$ is to move to some point in the
final body, $x_1 = x_2 = \dots = x_T = x^* := \proj_{K_T}(x_0)$, and
hence the optimal value is
$C^* = \dist(x_0,K_T) := \min_{x \in K_T} \|x_0 - x\|$.

Bansal et al.~\cite{Bansal} gave a $f(d)$-competitive algorithm for
chasing nested convex bodies in Euclidean space with
$f(d) = 6^d(d!)^2$. Our main theorem improves on their result:
\begin{theorem}[Main Theorem]
  \label{thm:main}
  For any norm, there is an $O(d \log d)$-competitive algorithm for nested \cbc.
\end{theorem}
This result is almost tight unless we
make further assumptions on the norm, since there is an $\Omega(d)$
lower bound for the $\ell_\infty$-norm~\cite{FL93}.

\paragraph{Our Technique.} The high-level idea behind our algorithm is
simple: we would like to stay ``deep'' inside the feasible region so
that when this point becomes infeasible, 
the feasible
region 
shrinks
considerably. One natural candidate for such a ``deep'' point is the
\emph{centroid}, i.e., the \emph{center of mass} of the feasible
region. Indeed, it is known by a theorem of Gr\"{u}nbaum that any
hyperplane passing through the centroid of a convex body splits it into
two pieces each containing a constant fraction of the volume. At first,
this seems promising, since after $O(d)$ steps the volume would drop by
$2^{d}$---now if the feasible region is well-rounded then it would halve
the diameter and we would have made progress. The problem is that the
feasible region may have some skinny directions and other fat ones, so
that the diameter may not have shrunk even though the volume has
dropped. Indeed, Bansal et al.~\cite{Bansal} give examples showing that
this na\"{\i}ve centroid algorithm --- and a related Ellipsoid-based
algorithm --- is not competitive.  While our algorithm is also based on
the centroid, it avoids the pitfalls illustrated by these examples.

Our main idea is that if we have a very skinny dimension---say the body
started off looking like a sphere and now looks like a pancake---we have
essentially lost a dimension! The body can be thought of as lying in a space with
one fewer dimension, and we should act accordingly.  Slightly more precisely,
we restrict to the skinny directions and solve the problem in
that subspace recursively. Our cost is tiny because these directions are
skinny. Once there are no points in this subspace, we can find a
hyperplane that cuts along the fat directions (i.e., parallel
to the
skinny directions), which makes progress towards reducing the diameter.

\paragraph{The Greedy Algorithm.} 
We also bound the competitive ratio of the simplest algorithm for this
problem, namely the \emph{greedy algorithm}.  This algorithm, at time
$t$, outputs the point $x_t = \proj_{K_t}(x_{t-1})$ obtained by moving
to the closest feasible point at each step. We observe that a slightly
better result than the Bansal et al.\ result can be obtained for this
algorithm as well.

The analysis of the greedy algorithm follows from a result of Manselli and 
Pucci~\cite{ManP} on self-contracted curves. 
A rectifiable curve $\gamma: [0,a]\to \R^d$ is \emph{self-contracted} 
if for every $x\in [0,a]$ such that $\gamma$ has a tangent vector 
$t(x)$ at $x$, the sub-curve $\gamma([x,a])$ is contained in the half-
space $\{y\in \R^d: \la t(x), \gamma(y)-\gamma(x)\ra \ge 0\}$.
\begin{theorem}[Manselli and Pucci~\cite{ManP}]\label{thm:manselli}
Let $\gamma$ be a self-contracted curve in $\R^d$ and $K$ be a bounded 
convex set containing $\gamma$. 
Then the length of $\gamma$ is at most $((d-1)d^{d/2}\cdot \frac{\omega_d}{\omega_{d-1}}) 
\cdot diam(K)$, where $\omega_d$ denotes the $(d-1)$-dimensional surface 
area of the $d$-sphere.
\end{theorem}

\begin{theorem}[Greedy]\label{thm:greedy}
  The greedy algorithm is $O(d^{(d+1)/2})$-competitive for nested \cbc with the Euclidean norm. 
\end{theorem}
\begin{proof} 
Let $(x_t)_{t \le T}$ be the piecewise affine extension of the greedy
  algorithm's path $(x_t)_{t \in [T]}$. Essentially by definition $(x_t)$
  is a {\em self-contracted curve}. For each point in $y\in K_t$, the distance $d(y,x_t)$ 
  is decreasing. In particular, if $y$ is the optimal solution with cost $OPT = \|y\|$, 
  $(x_t)$ is contained in the ball of radius $OPT$ centered at $y$.
  Theorem~\ref{thm:manselli} implies a competitiveness of $(d-1)d^{d/2}\cdot
  \frac{\omega_d}{\omega_{d-1}}$. Now using that the
  last ratio is at most $O(1/\sqrt{d})$ gives the claim.
\end{proof}

In upcoming joint work by the second author with O.\ Angel and F.\
Nazarov, we also show that greedy's competitive ratio is $\Omega(c^d)$
for some $c>1$ (and $O(C^d)$ for some $C >1$). Thus Theorem
\ref{thm:main} gives a provably exponential improvement over the greedy algorithm.

\subsection{Other Related Work}

Motivated by problems in power management in data servers, Andrew et
al.~\cite{LinWRMA12,AndrewBLLMRW13} studied the \emph{smoothed online
  convex optimization} (SOCO) problem (not to be confused with \emph{online convex optimization} in online learning).
SOCO (also called {\em chasing convex functions}) is a special case of MTS, where the
metric space is $\R^d$ equipped with some norm, and the cost functions
are convex; hence it generalizes \cbc. An $O(1)$-competitive algorithm
is known for $d=1$~\cite{AndrewBLLMRW13, BGKPS}. Antoniadis et
al.~\cite{Anto} gave an intuitive algorithm for chasing lines and affine
subspaces, and they also showed reductions between \cbc and
SOCO. Finally the online primal-dual framework~\cite{BN-mono} can also
be viewed as chasing nested \emph{covering} constraints, i.e.,
$K_t := \{ x \in \R^d : a_s^\intercal x \geq 1 \; \forall s \leq t\}$
for $a_s \geq 0$,
with the $\ell_1$ metric (i.e., $\|\cdot\| = \|\cdot\|_1$).

\subsection{Reductions}
\label{sec:reductions}

We recall some simple reductions between \cbc and two other closely
related problems, which allow for a guess-and-double approach. This
allows us to move between the original \cbc problem and its variants
where one plays in a convex body until its diameter falls by a
constant factor.

\begin{claim} \label{claim:reduction}
For some function $f(d) \geq 1$, the following three propositions are equivalent:
  \begin{itemize} \itemsep 1pt
  \item[(i)] (General) There exists a $f(d)$-competitive algorithm for nested \cbc.
  \item[(ii)] ($r$-Bounded) For any $r \geq 0$, and assuming that
  $K_1 \sse B(x_0,r)$, there exists an algorithm for nested \cbc with total movement 
  $O(f(d)\cdot r)$. 
\item[(iii)] ($r$-Tightening) For any $r \geq 0$, and assuming that
  $K_1 \sse B(x_0,r)$, there exists an algorithm for nested \cbc that
  incurs total movement cost $O(f(d) \cdot r)$ until the first time $t$
  at which 
  $K_t$ is
  contained in some ball of radius $\frac{r}{2}$.
  \end{itemize}
\end{claim}
\begin{proof}
  The implications $(i)\Rightarrow (ii)\Rightarrow (iii)$ are clear. To
  get $(ii)\Rightarrow (i)$, we iteratively run the $r$-Bounded
  algorithm starting at $x_0$ and doubling $r$ in each
  run. Specifically, let $r_0 := \dist(x_0,K_1)$ be the distance from
  the starting point to the first convex body $K_1$; this will be our
  initial guess $r_0$ for the optimal cost. (Without loss of generality,
  we can assume that $x_0 \not\in K_1$ and hence $r_0 \neq 0$, else we
  can drop $K_1$ from the sequence.) In the $k^{th}$ run, we execute the
  $r$-Bounded problem algorithm with parameter $r_k:=2^k r_0$ on the
  \emph{truncated} sets $B(x_0,r_k)\cap K_t$. If
  $B(x_0,r_k)\cap K_t = \emptyset$ for some $t$, then we know that
  $\dist(x_0, K_t) \geq r_k$, and hence the optimal cost is strictly
  more than $r_k$. In that case we move back to $x_0$, and begin the
  $k+1^{st}$ run of our algorithm (i.e., start playing the
  $r_{k+1}$-Bounded problem on the current truncated set
  $K_t \cap B(x_0, r_{k+1})$).
    
  By our assumption, the algorithm's cost for run $k$ 
  is $O(f(d) \cdot r_k)$, plus $r_k$ for the final
  cost of moving back to $x_0$ at the end of the iteration.  If the
  algorithm requires $T$ iterations, the optimal cost is at least
  $r_T/2$ (because the $(T-1)^{st}$ run ended) whereas our total cost is at
  most
  $$\sum_{k=1}^T O(f(d)+1))\, r_k = \sum_{k=1}^T O(f(d)+1))\; 2^{k-T} r_T 
  < 2 \cdot O(f(d)+1)\; r_T,$$ which implies a competitive ratio of $O(f(d))$.
  (This reduction $(ii)\Rightarrow (i)$ also appears
  as~\cite[Lemma~6]{Bansal}.)

  Finally, to show $(iii)\Rightarrow (ii)$, we first run the $r$-Tightening
  algorithm until 
  $K_t$ is
  contained in some ball $B(x', r/2)$ of radius $\nicefrac{r}{2}$, 
  then move to the center
  $x'$ of the smaller ball and use the $r/2$-Tightening algorithm on
  that ball, and so on. 
  The new
  center $x'$ is at distance $O(r)$ from the old center $x$, and hence
  the cost for the first ball is $O(f(d) + 1) \cdot r$. Now the cost is
  reduced by a factor of $\frac{1}{2}$ for each successive iteration,
  thus the total cost is at most $2 \cdot O(f(d)+1)\,r = O(f(d))\,r$.
\end{proof}

\input{convexgeometry}
\input{centroid}

\medskip\textbf{Acknowledgments.} We thank Nikhil
Bansal, Niv Buchbinder, Guru Guruganesh, and Kirk Pruhs for useful
conversations. C.J.A.\ and A.G.\ thank Sunny Gakhar
for discussions in the initial stages of this work. This work was supported in part
by NSF awards CCF-1536002, CCF-1540541, CCF-1617790, and
CCF-1740551, and the Indo-US Virtual Networked Joint Center on Algorithms under Uncertainty.

{\small 
\bibliography{bib}
\bibliographystyle{amsalpha}
}

\end{document}

%% file: def.tex
\def\E{\mathbb{E}}

\def\N{\mathbb{N}}

\def\R{\mathbb{R}}

\def\bf{\mathbf{f}}

\def\la{\langle}
\def\ra{\rangle}

\def\proj{\text{proj}}

\newcommand{\Vol}{\operatorname{Vol}}

\newcommand{\norm}[1]{\left\| #1 \right\|}


%% file: convexgeometry.tex
\section{Preliminaries}
We gather here notation and classical convex geometry results that will be useful in our analysis.

\subsection{Notation}
\label{sec:notation}

Given a convex body $K \sse \R^d$, its \emph{centroid} (also called its center
of mass/gravity) is 
\[ \mu(K) := \frac{1}{\Vol(K)} \int_{x \in K} x\, dx = \E_{X \sim K}[X]. \] Given a unit vector
$v \in \R^d$, the \emph{directional width} of a set $K$ in the direction $v$ is
\begin{gather}
  w(K, v) = \max_{x, y \in K} v^\intercal (x-y). \label{eq:3}
\end{gather} 
We denote $\delta(K)$ for the minimum directional width of $K$ over all unit vectors $v$. 
We define $\Pi_{L}X$ to be the projection of the set X on the subspace $L$, that is
$$\Pi_{L}X \defeq \{x+y \in L : x\in X, y \in L^\perp\}.$$
In the following we fix a norm $\|\cdot\|$ in $\R^d$, and use $B(x,
r) :=\{y \in \R^d : \|y-x\| \leq r\}$ to denote  a ball of radius $r \geq 0$ centered at $x \in \R^d$. Furthermore we will assume that the norm satisfies for all $x \in \R^d$:
\begin{equation} \label{eq:normass}
\|x\|_2 \leq \|x\| \leq \sqrt{d} \|x\|_2 ~.
\end{equation}
Indeed, given a full-dimension convex body $K\sse \R^d$ which is
symmetric about the origin, John's theorem guarantees the existence of
an ellipsoid $\mathcal{E}$ such that $\mathcal{E}\sse K \sse
\sqrt{d}\mathcal{E}$ (see, e.g.~\cite{Ball92}). Take $K$ to be the unit $\|\cdot \|$ ball, and by applying a linear transformation we may assume that $\mathcal{E}$ is the unit Euclidean ball. We see that \eqref{eq:normass} holds true, so we make this assumption without loss of generality. 

\subsection{Convex geometry reminders}
\label{sec:conv-geom}

We use the following theorems from convex geometry in our analysis. Let
$K$ denote a general convex body $K \subset \R^d$. Some definitions used
here were given in Section~\ref{sec:notation}.

\begin{theorem}[Gr\"unbaum's Theorem~\cite{Grunbaum}]
  \label{thm:Grunbaum} For any half-space $H$ containing the centroid of
  $K$ one has
  \[
    \Vol(K\cap H)\geq\frac{1}{e}\Vol(K).
  \]
\end{theorem}

\begin{theorem}[{\cite[Lemma~6.1]{LLV}}]
  \label{thm:projected_vol}
  Recall that $\delta(K)$ is the minimum directional width of $K$ over
  all directions.  For any subspace $L \subset \R^d$ one has
  \[
    \Vol(\Pi_{L}K)\leq\left(\frac{d(d+1)}{\delta(K)}\right)^{d-\dim L}\cdot\Vol(K).
  \]
\end{theorem}

\begin{theorem}[{\cite[Theorem~4.1]{kannan1995isoperimetric}}]
  \label{thm:sandwish} Let $\mu$ denote the centroid of $K$. Let
  $\Sigma \defeq\E_{x\sim K}(x-\mu)(x-\mu)^{T}$ be the covariance matrix of
  $K$, and let
  $\mathcal{E} \defeq \{x \in \R^d : x^{\top} \Sigma^{-1} x\leq1\}$ be the
  ellipsoid defined by $\Sigma$.  Then
  one has
  \[ \mu+\sqrt{\frac{d+1}{d}}\mathcal{E}\subset
  K\subset\mu+\sqrt{d(d+1)}\mathcal{E}. \]
\end{theorem}

\begin{lemma}
  \label{lem:width_after_cut}
  Any convex body $K \sse \R^d$ contains a Euclidean ball of radius
  $\delta(K)/d$. Furthermore, for any halfspace $H$ containing the
  centroid $\mu$ on its boundary (i.e.,
  $H=\{x \in \R^d :v^{\top}(x-\mu)\leq0\}$ for some $v$), one has
  \[ \delta(H \cap K) \geq \delta(K) / (2d). \]
\end{lemma}
\begin{proof}
  Let $\mathcal{E}$ be the ellipsoid defined in Theorem
  \ref{thm:sandwish}. Since $K\subset\mu+\sqrt{d(d+1)}\mathcal{E}$ we
  have that the minimum width of the scaled ellipsoid
  $\sqrt{d(d+1)}\mathcal{E}$ is larger than $\delta(K)$, which in turn
  implies that $\mathcal{E}$ contains a Euclidean ball of radius
  $\frac{\delta(K)}{\sqrt{d(d+1)}}$ centered at $\mu$. Thus using that
  $\mu+\sqrt{\frac{d+1}{d}}\mathcal{E}\subset K$ we get that $K$
  contains a Euclidean ball of radius $\delta(K)/d$ centered at
  $\mu$. The second statement follows since a half-ball of radius $r$
  contains a ball of radius $r/2$.
\end{proof}


%% file: centroid.tex
\section{A Centroid-Based Algorithm}
\label{sec:centroid}

We present a centroid-based algorithm and the analysis that it is
$f(d) = O(d\log d)$ competitive. By the reductions in
Claim~\ref{claim:reduction} and by scaling, it suffices to give an
algorithm for the $1$-Tightening version problem, which starts off with
the convex body being contained in a unit ball, terminates with the
final body lying in a ball of radius $1/2$, and pays at most $f(d)$.

\subsection{Overview of the \Bounded and \Tighten Algorithms}
To simplify notation we ignore the time indexing $t$, and we denote $\curr$ for the 
{\em current} requested convex body. Thus at the start of the algorithm one has $\curr=K_1$,
and every time the algorithm moves to a new point $x$ in $\curr$ the adversary updates 
$\curr$ to the next convex body in the input sequence that does not contain $x$.

Both the \Bounded and \Tighten algorithms take as input an affine subspace $A\sse \R^d$. 
They each output points $x_t\in \curr\cap A$ until their respective end conditions are met. The two algorithms only differ essentially in their end conditions:
\Bounded terminates when $\curr\cap A$ is empty while \Tighten terminates when $\curr\cap A$ is ``skinny" in every direction.

In the next section we will define the \Tighten algorithm. Recall that using the reductions of Section~\ref{sec:reductions}, we then also get an algorithm for the \Bounded version of the problem. The algorithm \Tighten makes calls to \Bounded only in lower-dimensional subspaces so there is no circular reference.

\subsection{The \Tighten Algorithm}
Due to the recursive nature of our algorithm we will consider solving the \Tighten problem
in an affine subspace $A \subset \R^d$. More precisely, while $\curr$ is a
convex body in $\R^d$, we are only interested in its ``shadow'' $\currA := \curr \cap A$.
Given the guarantee that at the start $\currA$ is contained in a unit ball (in $A$), the 
goal is to output points within $\currA$ until $\currA$ lies inside some
ball of radius at most $1/2$.

\bigskip
\hrule width \hsize height 2pt \kern 1mm
\hrule width \hsize 
\smallskip

\textbf{Algorithm} $\Tighten(A)$

\medskip
\hrule

\begin{enumerate}
\item Let $d_A$ be the dimension of $A$; if $d_A \leq 1$ run the greedy algorithm.
\item Let $k=0$ and $\delta = \frac{c}{d_A^{3}\log d_A}$ for some small enough
constant $c$.
\item Let $S_{0} \sse A$ be a subspace of ``skinny'' directions obtained
  by choosing a maximal set of orthogonal directions $v$ such that the
  directional width is smaller than $\delta$, i.e., $w(\currA, v) \leq \delta$.
  
\item While $S_{k}\neq A$ 
\begin{enumerate}
\item 
  Let $c_k$ be the centroid of $\Pi_{S_{k}^{\perp}} \currA$. Move to any point $x_k\in \currA$ such that $\Pi_{S_k^\perp} x_k = c_k$. 
\item Call the procedure $\Bounded(x_k+S_k)$.
\item Let $S \gets S_k$. While there is a direction $v\perp S$ such that $w(\currA, v) \leq \delta$, $S \gets \text{span}(v,S)$.
\item $S_{k+1} \gets S$, and $k \gets k+1$.
\end{enumerate}
\end{enumerate}

\medskip
\hrule width \hsize \kern 1mm
\hrule width \hsize height 2pt 
\bigskip

Observe that when $\Tighten(A)$ stops, the body $\currA$ has directional width at most $\delta$ for any direction in some basis of $A$ (since $S_k = A$), and thus it is contained in a Euclidean ball of radius $O(\sqrt{d} \delta)$.
This means that the diameter in the norm $\|\cdot\|$ is at most $O(d \delta) \leq 1/2$ (recall~\eqref{eq:normass}).
In particular this is a valid stopping time for the $\Tighten$ problem. Thus we only have to analyze the movement cost of $\Tighten(A)$.

\paragraph{An Example.}
Suppose $A = \R^3$ at the beginning of some iteration $k$,
$\currA := \{(x,y,z) \mid x^2 + y^2 \leq 1, z \in [-\delta,\delta]\}$ is
a unit-radius pancake with height $\delta$ centered at the origin, with
the major axes in the $x$-$y$ plane and the short dimension along the
$z$-axis. The subspace $S_k \sse A$ is such that $\currA$ is skinny
along directions in this subspace, in this case suppose $S_k$ is the
$z$-direction.  The algorithm takes the projection of $\currA$ onto the
non-skinny directions (the $x$-$y$ plane), finds its centroid $c_k$ and
chooses $x_k \in \currA$ that projects to this centroid $c_k$. In this
case the centroid of the projection is the origin, and then $x_k$ is any
point in the pancake with $x=y=0$.

The algorithm then recurses with $A$
being the $z$-axis, i.e., the subspace $x=y=0$. When this recursive call
terminates, $\currA$ has an empty intersection with this affine subspace
(i.e., with the $z$-axis). This means there exists a hyperplane that
separates the $z$-axis from the new $\currA$---in particular, the new $\currA$
lies within some ``half-pancake''. This operation not only reduces $\currA$'s
volume, but also makes a substantial reduction in its width along some direction in
the $x$-$y$ plane.

\subsection{Cost Analysis}
Each iteration of $\Tighten(A)$ induces a movement of at most $1$ when
we move to the centroid (recall that by assumption $\currA$ is of diameter at most $1$)
plus the movement in the recursive call. The latter movement is tiny, since in the recursion the body
has a small diameter ($\leq \delta$). So the main part of the analysis
is to bound the number of iterations; the bound on the total movement is
then proved in Theorem~\ref{thm:move}.

In the following we denote $X_k$ for the value of $\currA$ at the beginning of iteration $k$.

\begin{lemma}
  \label{lem:iters}
  The procedure \Tighten terminates in at most $O(d \log d)$ iterations.
\end{lemma}

\begin{proof}
    We bound the number of iterations of
  the algorithm via the potential
  \[ \Vol(\Pi_{S_{k}^{\perp}}X_{k}). \] At a high level, the proof shows that
  the hyperplane cuts cause this projected volume to decrease rapidly
  (since we are making cuts along the non-skinny directions). And when
  $S_k$ grows, we show that the projected volume does not increase too
  much.

  The key observation is that
  $\Pi_{S_{k}^{\perp}}X_{k+1}$ is contained in the intersection of
  $\Pi_{S_{k}^{\perp}}X_{k}$ with a halfspace $H_k$ passing through its centroid $c_k$.
  Indeed after the recursive call to $\Bounded$ in iteration $k$, one
  has that $\currA \cap (x_k + S_k) = \emptyset$; recall that
    $\currA=X_{k+1}$ at that time.
  But $c_k \not\in \Pi_{S_k^\perp} X_{k+1}$ implies that
  $\Pi_{S_{k}^{\perp}}X_{k+1}$ and $c_k$ can be separated by some hyperplane.   
  
  Let $Y_k$ denote the intersection $\Pi_{S_{k}^{\perp}}X_{k}\cap H_k$.
  By Gr\"unbaum's theorem
  (Theorem~\ref{thm:Grunbaum}),
  \begin{equation}
    \Vol(Y_k) \leq\left(1-\frac{1}{e}\right)\Vol(\Pi_{S_{k}^{\perp}}X_{k}).\label{eq:vol_1}
  \end{equation}

  Now the construction of $S_k$ ensures that there are no directions
  orthogonal to $S_k$ that are skinny. Hence the minimum width of
  $\Pi_{S_{k}^{\perp}}X_{k}$ is at least $\delta$. By
  Lemma~\ref{lem:width_after_cut}, the minimum width of
  $Y_k$ is at least $\delta/2d$. Now, we 
  apply Theorem \ref{thm:projected_vol} to $Y_k$ and use to get:
  \begin{equation}
	\Vol(\Pi_{S_{k+1}^\perp}Y_k)\leq
    \left(\frac{2d^{2}(d+1)}{\delta}\right)^{\dim
    S_{k+1}-\dim S_{k}}\Vol(Y_k).\label{eq:vol_2}
  \end{equation}
	The containments $\Pi_{S_{k}^{\perp}}X_{k+1}\sse Y_k$ and $S_{k+1}^\perp \sse S_k^\perp$ imply:
  \begin{equation}
	\Vol(\Pi_{S_{k+1}^\perp}X_k)\leq \Vol(\Pi_{S_{k+1}^\perp}Y_k)
	\label{eq:vol_3}
  \end{equation}

  Combining (\ref{eq:vol_1}), (\ref{eq:vol_2}), and (\ref{eq:vol_3}),
  we have that after $T$ steps,
  \begin{align}
    \Vol(\Pi_{S_{T}^{\perp}}X_{T})
    &
      \leq\left(\frac{2d^{2}(d+1)}{\delta}\right)^{\dim
      S_{T}} \cdot
      \left(1-\frac{1}{e}\right)^{T}
      \cdot \Vol(X_{0})\nonumber \\
    &
      \leq\left(\frac{2d^{2}(d+1)}{\delta}\right)^{d}
      \cdot
      \left(1-\frac{1}{e}\right)^{T} \cdot
      O(1)^{d}\label{eq:volT}
  \end{align}
  where we used that $X_{0}$ is contained in the unit ball in
  $\norm{\cdot}_{2}$ and hence has volume at most $O(1)^{O(d)}$.

  Now let $T$ be the last iteration where $S_T^\perp$ has non-zero
  number of dimensions (i.e., the step just before the procedure
  ends). At this point $\Pi_{S_{T}^{\perp}}X_{T-1}$ has minimum width at
  least $\delta$.  Lemma \ref{lem:width_after_cut} shows that it
  contains a ball of radius $\delta/2d$, and hence has volume at least
  $(\delta/d)^{O(d)}$. This gives a lower bound on the potential.

  Combining this lower bound with \eqref{eq:volT} which upper bounds the
  potential after $T$ steps, we have that
  $T=O(d (\log d + \log\nicefrac{1}{\delta}))=O(d\log d)$ using our
  choice of $\delta$.
\end{proof}

\begin{theorem}
  \label{thm:move}
  The total movement of the \Tighten procedure is $O(d \log d)$,
  assuming that at the start $\currA$ is contained in a unit ball.
\end{theorem}

\begin{proof}
  We induct on the number of dimensions $d$: the base case is when $d=0$
  or $d=1$, in which case the claim is immediate. Hence consider $d\geq2$.  
  Let $f_{T}(d)$ be the total movement of the algorithm
  $\Tighten$ and $f_{B}(d)$ be the total movement of the
  algorithm $\Bounded$. By the reduction in Claim
  \ref{claim:reduction}, 
  \[
    f_{B}(d)\leq O(f_{T}(d)).
  \]
  Next, in each iteration, the diameter of $\currA \cap(x_{k}+S_{k})$ is
  at most $d\delta$ in $\norm{\cdot}$ (by the triangle inequality)
  Therefore, the cost per each iteration
  is at most
  \[
    O(1)+ O(d\delta\cdot f_{B}(d-1))=O(1+d\delta\cdot f_{T}(d-1)),
  \]
  where the first term of $O(1)$ comes from the fact that $\currA$ is
  shrinking and hence is always contained in an unit ball. The second
  term comes from the recursion on a body of diameter $d \delta$ in
  $d-1$ dimensions. In the second term, we bound $f_B(d-1)$ by $O(f_T(d-1))$.

  Since there are at most $O(d\log d)$ iterations from
  Lemma~\ref{lem:iters}, 
  \begin{align*}
    f_{T}(d) & \leq O(d\log d) \cdot O(1+d\delta f_{T}(d-1))\\
             & =O(d\log d+d^{2}\log d \cdot \delta\cdot f_{T}(d-1)).
  \end{align*}
  We choose $\delta\leq\frac{c}{d^{3}\log d}$ for a small enough
  constant $c$, and use $f_T(d-1) = O(d \log (d-1))$ from the inductive
  hypothesis. This gives
  $f_T(d) = O(d \log d + \frac{1}{d}\cdot d \log( d-1)) = O((d+1) \log
  d)$, which proves the theorem.
\end{proof}


%% file: main.bbl
\newcommand{\etalchar}[1]{$^{#1}$}
\providecommand{\bysame}{\leavevmode\hbox to3em{\hrulefill}\thinspace}
\providecommand{\MR}{\relax\ifhmode\unskip\space\fi MR }
\providecommand{\MRhref}[2]{%
  \href{http://www.ams.org/mathscinet-getitem?mr=#1}{#2}
}
\providecommand{\href}[2]{#2}
\begin{thebibliography}{LWR{\etalchar{+}}12}

\bibitem[ABL{\etalchar{+}}13]{AndrewBLLMRW13}
Lachlan L.~H. Andrew, Siddharth Barman, Katrina Ligett, Minghong Lin, Adam
  Meyerson, Alan Roytman, and Adam Wierman, \emph{A tale of two metrics:
  Simultaneous bounds on competitiveness and regret}, {COLT} 2013 - The 26th
  Annual Conference on Learning Theory, June 12-14, 2013, Princeton University,
  NJ, {USA}, 2013, pp.~741--763.

\bibitem[ABN{\etalchar{+}}16]{Anto}
Antonios Antoniadis, Neal Barcelo, Michael Nugent, Kirk Pruhs, Kevin Schewior,
  and Michele Scquizzato, \emph{Chasing convex bodies and functions}, L{ATIN}
  2016: theoretical informatics, Lecture Notes in Comput. Sci., vol. 9644,
  Springer, Berlin, 2016, pp.~68--81. \MR{3492519}

\bibitem[Bal92]{Ball92}
Keith Ball, \emph{Ellipsoids of maximal volume in convex bodies}, Geom.
  Dedicata \textbf{41} (1992), no.~2, 241--250. \MR{1153987}

\bibitem[BBE{\etalchar{+}}17]{Bansal}
Nikhil Bansal, Martin B\"ohm, Marek Eli{\'{a}}\v{s}, Grigorios Koumoutsos, and
  Seeun~William Umboh, \emph{Nested convex bodies are chaseable}, Proceedings
  of the ACM-SIAM symposium on Discrete algorithms (SODA 2018) (2017).

\bibitem[BGK{\etalchar{+}}15]{BGKPS}
Nikhil Bansal, Anupam Gupta, Ravishankar Krishnaswamy, Kirk Pruhs, Kevin
  Schewior, and Cliff Stein, \emph{A 2-competitive algorithm for online convex
  optimization with switching costs}, APPROX, vol.~40, Schloss Dagstuhl., 2015,
  pp.~96--109. \MR{3441958}

\bibitem[BLS92]{BRS}
Allan Borodin, Nathan Linial, and Michael~E. Saks, \emph{An optimal on-line
  algorithm for metrical task system}, J. Assoc. Comput. Mach. \textbf{39}
  (1992), no.~4, 745--763. \MR{1187210}

\bibitem[BN07]{BN-mono}
Niv Buchbinder and Joseph~(Seffi) Naor, \emph{The design of competitive online
  algorithms via a primal-dual approach}, Found. Trends Theor. Comput. Sci.
  \textbf{3} (2007), no.~2-3, 93--263. \MR{2506496 (2010h:68239)}

\bibitem[FL93]{FL93}
Joel Friedman and Nathan Linial, \emph{On convex body chasing}, Discrete
  Comput. Geom. \textbf{9} (1993), no.~3, 293--321. \MR{1204785}

\bibitem[Gr{\"u}60]{Grunbaum}
B.~Gr{\"u}nbaum, \emph{Partitions of mass-distributions and of convex bodies by
  hyperplanes}, Pacific J. Math. \textbf{10} (1960), 1257--1261. \MR{0124818}

\bibitem[KLS95]{kannan1995isoperimetric}
Ravi Kannan, L{\'a}szl{\'o} Lov{\'a}sz, and Mikl{\'o}s Simonovits,
  \emph{Isoperimetric problems for convex bodies and a localization lemma},
  Discrete \& Computational Geometry \textbf{13} (1995), no.~1, 541--559.

\bibitem[LLV17]{LLV}
Ilan Lobel, Renato~Paes Leme, and Adrian Vladu, \emph{Multidimensional binary
  search for contextual decision-making}, Proceedings of the 2017 {ACM}
  Conference on Economics and Computation, {EC} '17, Cambridge, MA, USA, June
  26-30, 2017, 2017, p.~585.

\bibitem[LWR{\etalchar{+}}12]{LinWRMA12}
Minghong Lin, Adam Wierman, Alan Roytman, Adam Meyerson, and Lachlan L.~H.
  Andrew, \emph{Online optimization with switching cost}, {SIGMETRICS}
  Performance Evaluation Review \textbf{40} (2012), no.~3, 98--100.

\bibitem[MP91]{ManP}
Paolo Manselli and Carlo Pucci, \emph{Maximum length of steepest descent curves
  for quasi-convex functions}, Geom. Dedicata \textbf{38} (1991), no.~2,
  211--227. \MR{1104346}

\bibitem[Sit14]{Sitters}
Ren\'e Sitters, \emph{The generalized work function algorithm is competitive
  for the generalized 2-server problem}, SIAM J. Comput. \textbf{43} (2014),
  no.~1, 96--125. \MR{3158795}

\end{thebibliography}
